\documentclass{my-jgaa-art}

\usepackage{graphicx}
\usepackage{tabularx}
\usepackage{amsmath}
\usepackage{amssymb}

\usepackage[utf8]{inputenc}
\usepackage[T1,T2A]{fontenc} %
\usepackage{amsmath,amssymb,amsfonts}
\usepackage{latexsym}
\usepackage{xspace}
\usepackage{colortbl,xcolor}
\usepackage[misc,geometry]{ifsym}
\usepackage{mathtools}
\usepackage[inline]{enumitem}
\usepackage{todonotes}
\usepackage{url}
\usepackage{cite}
\usepackage[capitalise,nameinlink]{cleveref}
\newtheorem{observation}[theorem]{Observation}

\newtheorem{proposition}[theorem]{Proposition}

\crefname{proposition}{Proposition}{Propositions}
\crefname{observation}{Observation}{Observations} %
\newcommand{\proofparagraph}[1]{\par\addvspace\topsep\noindent\emph{#1}\;}

\graphicspath{{./Figures/}}
\def\id{\operatorname{id}}
\def\diff{\operatorname{diff}}

\definecolor{defblue}{rgb}{0.121,0.47,0.705}
\definecolor{linkblue}{rgb}{0.098,0.098,0.4392}
\let\emph\relax
\DeclareTextFontCommand{\emph}{\color{defblue}\em}

\newcommand{\tangleMinimization}{\textsc{Tangle-Height Minimization}\xspace}
\newcommand{\listFeasibility}{\textsc{List-Feasibility}\xspace}
\newcommand{\oh}{\ensuremath{\mathcal{O}}}

\begin{document}

\HeadingAuthor{O.~Firman, P.~Kindermann, B.~Klemz,
  A.~Ravsky, A.~Wolff, and J.~Zink}
\HeadingTitle{Deciding the Feasibility and Minimizing the
  Height of Tangles} %
\title{\renewcommand*{\thefootnote}{\fnsymbol{footnote}} %
  Deciding the Feasibility and \mbox{Minimizing the
  Height of Tangles}\footnote{Preliminary
    versions of this work have appeared in Proc.\ 27th
    Int.\ Symp.\ Graph Drawing \& Network
    Vis.\ (GD 2019) \cite{fkrwz-cohtf-GD19} and Proc. 49th
    Int.\ Conf.\ Current Trends Theory \& Practice
    Comput.\ Sci.\ (SOFSEM 2023) \cite{fkkrwz-cotf-SOFSEM23}}}

\authorOrcid[1]{Oksana~Firman}{oksana.firman@uni-wuerzburg.de}{0000-0002-9450-7640} %
\authorOrcid[2]{Philipp~Kindermann}{kindermann@uni-trier.de}{0000-0001-5764-7719} %
\authorOrcid[1]{Boris Klemz}{boris.klemz@uni-wuerzburg.de}{0000-0002-4532-3765} %
\author[3]{Alexander~Ravsky}{alexander.ravsky@uni-wuerzburg.de} %
\authorOrcid[1]{Alexander~Wolff}{}{0000-0001-5872-718X} %
\authorOrcid[1]{Johannes~Zink}{johannes.zink@uni-wuerzburg.de}{0000-0002-7398-718X} %

\affiliation[1]{Universität Würzburg, Würzburg, Germany} %
\affiliation[2]{Universität Trier, Trier, Germany}
\affiliation[3]{Pidstryhach Institute for Applied Problems of
  Mechanics and Mathematics,\\
  National Academy of Sciences of Ukraine, Lviv, Ukraine}

\maketitle

\begin{abstract}
  We study the following combinatorial problem. Given a set of $n$
  y-monotone \emph{wires}, a \emph{tangle} determines the order of the
  wires on a number of horizontal \emph{layers} such that the orders
  of the wires on any two consecutive layers differ only in swaps of
  neighboring wires.  Given a multiset~$L$ of \emph{swaps} (that is,
  unordered pairs of wires) and an initial order
  of the wires, a tangle \emph{realizes}~$L$ if each pair of wires
  changes its order exactly as many times as specified by~$L$.
  \listFeasibility is the problem of finding a tangle that
  realizes a given list~$L$ if such a tangle exists.
  \tangleMinimization is the problem of finding a
  tangle that realizes a given list and additionally uses the minimum
  number of layers.  \listFeasibility (and therefore
  \tangleMinimization) is NP-hard
  [Yamanaka et al., CCCG 2018].

  We prove that \listFeasibility remains NP-hard if every
  pair of wires swaps only a constant number of times.
  On the positive side, we present an algorithm for
  \tangleMinimization that computes an optimal tangle
  for $n$ wires and a given list~$L$ of swaps in
  $\oh((2|L|/n^2+1)^{n^2/2} \cdot \varphi^n \cdot n)$ time, where
  $\varphi \approx 1.618$ is the golden ratio
  and $|L|$ is the total number of swaps in~$L$.

  From this algorithm, we derive a simpler and faster version to solve
  \listFeasibility.  We also use the algorithm to show that
  \listFeasibility is in NP and fixed-parameter
  tractable with respect to the number of wires.
  For \emph{simple} lists, where every swap occurs at most once, we
  show how to solve \tangleMinimization in
  $\oh(n!\varphi^n)$ time.
\end{abstract}

\section{Introduction}\label{sec:intro}

This paper concerns the visualization of
\emph{chaotic attractors}, which occur in (chao\-tic) dynamic systems.
Such systems are considered in physics, celestial mechanics,
electronics, fractals theory, chemistry, biology, genetics, and
population dynamics; see, for instance,~\cite{cfps-c-SA86},~\cite{s-dgpd-SA02}, and~\cite[p.\ 191]{k-oosose-93}.
Birman and Williams~\cite{bw-kpods-Topology80}
were the first to mention tangles as a way to describe the topological
structure of chaotic attractors.  They investigated how the orbits of
attractors are knotted.  Later Mindlin et al.~\cite{mhsgt-csai-PRL90}
characterized attractors using integer matrices that
contain numbers of swaps between the orbits.

\renewcommand{\floatpagefraction}{.8}

Olszewski et al.~\cite{omkbrdb-vtca-GD18} studied the problem of
visualizing chaotic attractors.
In the framework of their paper,
one is given a set of $n$ wires that hang off a horizontal
line in a fixed order, and a multiset of swaps between the wires;
a tangle then is a visualization of these swaps, i.e., a sequence of
horizontal layers with an order of wires on each of the layers
such that the swaps are performed in the way that only adjacent wires can be swapped
and disjoint swaps can be done
simultaneously.
For examples of lists of swaps (described by a multilist and by an $(n \times n)$-matrix)
and tangles that realize these lists, see
\cref{fig:simple_ex,fig:many_loops}.

Olszewski et al.\ gave an
exponential-time algorithm for minimizing the \emph{height} of a
tangle, that is, the number of layers. We call this problem
\tangleMinimization. They tested their algorithm on
a benchmark set.

In an independent line of research, Yamanaka et
al.~\cite{yhuw-llr-CCCG18} showed that the problem
\textsc{Ladder-Lottery Realization} is NP-hard.  As it turns out, this
problem is equivalent to deciding the feasibility of a list, i.e.,
deciding whether there exists a tangle realizing the list.
We call this problem \listFeasibility.

Sado and Igarashi~\cite{si-fectbs-TCS87} used the same objective
function for \emph{simple} lists, that is, lists where each swap
appears at most once. (In their setting, instead of a list, the
start and final permutations are given but this
uniquely defines a simple list of swaps.)
They used odd-even sort, a parallel variant of bubblesort,
to compute tangles
with at most one layer more than the minimum.
Their algorithm runs in quadratic time.
Wang~\cite{w-nrsic-DAC91} showed that there is always a height-optimal
tangle where no swap occurs more than once. Bereg et
al.~\cite{bhnp-rpfm-SIAMJDM16,bhnp-dpfc-GD13} considered a similar
problem.  Given a final permutation, they showed how to minimize the
number of bends or \emph{moves} (which are maximal ``diagonal''
segments of the wires).

Let $L^1=(l^1_{ij})$ denote the (simple) list with $l^1_{ij}=1$
if $i\ne j$, and $l^1_{ij}=0$ otherwise.  This list is feasible
even if we start with any permutation of $\{1, \dots, n\}$;
a tangle realizing $L^1$ is commonly known as
\emph{pseudo-line arrangement}.  So tangles can be thought of as
generalizations of pseudo-line arrangements where the numbers of swaps
are prescribed and even feasibility becomes a difficult question.

In this paper we give new, faster algorithms for \listFeasibility and
\tangleMinimization, but before we can present our contribution in
detail, we need some notation.

\paragraph{Framework, Terminology, and Notation.}\label{FTN}

A \emph{permutation} is a bijection of the set $[n]=\{1,\dots,n\}$
onto itself.  The set $S_n$ of
all permutations of the set $[n]$ is a group whose multiplication is a composition of maps (i.e.,
$(\pi\sigma)(i)=\pi(\sigma(i))$ for each pair of permutations $\pi,\sigma\in S_n$ and each $i\in [n]$).
The identity of the group~$S_n$ is the identity permutation~$\id_n=\langle1,2,\dots,n\rangle$.
We write a permutation $\pi\in S_n$ as the sequence
of numbers $\pi^{-1}(1)\pi^{-1}(2)\dots\pi^{-1}(n)$.
In this sequence, element $i\in [n]$ is placed at the position~$\pi(i)$.
For instance, the permutation $\pi \in S_4$ with $\pi(1)=3$,
$\pi(2)=4$, $\pi(3)=2$, and $\pi(4)=1$ is written as $4312$.

Two permutations $\pi$ and $\sigma$ of $S_n$ are \emph{adjacent} if
they differ only in transposing neighboring elements, that is, if, for
every $i\in[n]$, $|\pi(i)-\sigma(i)|\le 1$. Then the set
$\{i\in [n]:\pi(i)\ne\sigma(i)\}$ splits into pairs $\{i,j\}$ such
that $\pi(i)=\sigma(j)$ and $\pi(j)=\sigma(i)$.

For $i,j \in [n]$ with $i \neq j$, applying the \emph{swap} $(i,j)$ to
a permutation~$\pi$ yields an adjacent permutation~$\sigma$ such that
$\pi(i)=\sigma(j)$ and $\pi(j)=\sigma(i)$.  For two adjacent
permutations~$\pi$ and~$\sigma$, let
\[\diff(\pi,\sigma)= \big\{ (i,j) \mid i,j \in[n] \; \land \;
i \ne j \; \land \; \pi(i)=\sigma(j) \; \land \; \pi(j)=\sigma(i) \big\}\]
be the set of
swaps in which~$\pi$ and~$\sigma$ differ. Given
a set of y-monotone curves called \emph{wires} that hang off a
horizontal line in a prescribed order~$\pi_1$, we define
a \emph{list} $L=(l_{ij})$ of \emph{order} $n$ to be a
symmetric $n\times n$ matrix with
entries in $\mathbb{N}_0$ and zero
diagonal. A list $L=(l_{ij})$ can also be considered
as a multiset of swaps,
where $l_{ij}$ is the multiplicity of swap~$(i,j)$.
By $(i,j) \in L$ we mean $l_{ij}>0$.

A \emph{tangle} $T$ of \emph{height}~$h$ realizing~$L$ is a sequence $\langle \pi_1,\pi_2,\dots,\pi_h \rangle$ of
permutations of the wires such that (i)~consecutive permutations are
adjacent and (ii)~$L=\bigcup_{i=1}^{h-1} \diff(\pi_i,\pi_{i+1})$.
Recall that $L$ is a multiset, so the union in~(ii) can yield several
copies of the same swap.
A \emph{subtangle} of~$T$ is a non-empty sequence
$\langle\pi_p, \pi_{p+1},\dots,\pi_q\rangle$
of consecutive permutations of~$T$ (that is, $1 \le p \le q \le h$).

\begin{figure}[thb]	
	~\hfill\hfill
	\includegraphics[page=1]{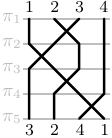}
	\hfill
	\includegraphics[page=2]{simple-example2}
	\hfill\hfill~
	
	\caption{Tangles $T$ and $T'$ of different heights realizing the list
		$L=\{(1,2), (1,3), (1,4), (2,3)\}$.}
	\label{fig:simple_ex}
\end{figure}

For example, the list~$L$ in \cref{fig:simple_ex} admits a tangle
realizing it.  We call such a list \emph{feasible}.
The list $L' = L \,\cup\, \{(1,2)\}$ with two $(1,2)$ swaps,
in contrast, is not feasible.  If the start permutation~$\pi_1$ is not
given explicitly, we assume that $\pi_1=\id_n$.  In
\cref{fig:many_loops}, the list $L_n$ is feasible; it is specified
by an $(n \times n)$-matrix.  The gray horizontal bars correspond to
the permutations (or \emph{layers}).

As a warm-up, we characterize the tangles that realize~$L_n$; this
characterization will be useful in \cref{sec:complexity}.

\begin{observation}
	\label{obs:unique}
	The tangle in \cref{fig:many_loops} realizes the list~$L_n$
	specified there; all tangles that realize~$L_n$ have the same
	order of swaps along each wire.
\end{observation}

\begin{proof}
	For $i, j \in [n - 2]$ with $i \neq j$, the wires $i$ and $j$
	swap exactly once, so their order reverses.
	Additionally, each wire $i \in [n - 2]$ swaps twice with
	the wire $k \in \{ n - 1, n\}$ that has the same parity as $i$.
	Observe that wire $i \in [n - 2]$ must first swap with each $j \in [n - 2]$
	with $j > i$, then twice with the correct $k \in \{ n - 1, n\}$,
	say $k = n$, and finally with each $j' \in [n - 2]$ with $j' < i$.
	Otherwise, there is some wire $i \in [n-1]$ that swaps with
	$i - 1$ before swapping with $n$.
	Then, $i$ cannot reach $n$ because wire $i - 1$ swaps only with wire $n - 1$
	among the two wires $\{n -1, n\}$ and thus separates $i$ from $n$.
	This establishes the unique order of swaps along each wire.
\end{proof}

\begin{figure}[tb]
	\begin{minipage}[c]{.48\linewidth}
		\centering
		
		$L_n=
		\begin{pmatrix}
		0 & 1 & 1 & \dots & 1 & \bf 0 & \bf 2 \\
		1 & 0 & 1 & \dots & 1 & \bf 2 & \bf 0 \\
		1 & 1 & 0 & \dots & 1 & \bf 0 & \bf 2 \\[-.8ex]
		\vdots & \vdots & \vdots & \ddots & \vdots & \vdots & \vdots \\
		1 & 1 & 1 & \dots & 0 & 0 & 2 \\
		\bf 0 & \bf 2 & \bf 0 & \dots & 0 & 0 & n-1 \\
		\bf 2 & \bf 0 & \bf 2 & \dots & 2 & n-1 & 0 \\
		\end{pmatrix}$
		
		\bigskip
		
		(The bold zeros and twos must be swapped if $n$ is even.)
	\end{minipage}
	\hfill
	\begin{minipage}[c]{.48\linewidth}
		\centering
		\includegraphics{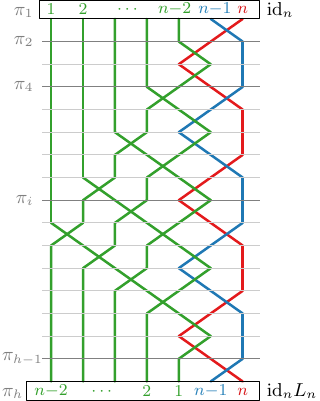}
	\end{minipage}
	
	\caption{A list $L_n$ for $n$ wires and a tangle of height $3n-4$
		realizing~$L_n$ (for $n=7$). The tangle height is minimum. %
	}
	\label{fig:many_loops}
\end{figure}

Let $|L|=\sum_{i<j} l_{ij}$ be the \emph{length} of~$L$.  A
list $L'=(l'_{ij})$ is a \emph{sublist} of $L$ if $l'_{ij}\le l_{ij}$
for each $i,j\in [n]$.
If there is a pair $(i',j') \in [n]^2$ such that $l'_{i'j'}<
l_{i'j'}$, then $L'$ is a \emph{strict} sublist of~$L$.
A list is \emph{simple} if all its entries are zeros or ones.
For any two lists $L = (l_{ij})$ and $L' = (l'_{ij})$
such that, for each $i,j \in [n]$, $l'_{ij}\le l_{ij}$,
let $L-L' = (l_{ij}-l'_{ij})$.

In order to understand the structure of feasible lists better,
we consider the following relation between them.
Let $L=(l_{ij})$ be a feasible list. Let $L'=(l'_{ij})$ be a list with
$l'_{ij}=l_{ij}$ for all $i, j \in [n]$ except for a single pair
$(i',j') \in [n]^2$, where $l'_{i'j'} = l_{i'j'} + 2$.
We claim that, if $l_{i'j'}>0$, then the list $L'$ is feasible, too.
To this end, note that any tangle $T$ that realizes $L$ has
two neighboring layers with adjacent permutations~$\pi$ and~$\pi'$
such that $(i,j) \in \diff(\pi,\pi')$.
Directly after~$\pi$,
we can insert two swaps $(i,j)$ into~$T$.  This yields a tangle
that realizes $L'$.
Given two lists $L=(l_{ij})$ and $L'=(l'_{ij})$, we write $L\to L'$ if
the list $L$ can be \emph{extended} to the list $L'$
by iteratively applying the above operation.

For a list $L=(l_{ij})$, let $1(L)=(l_{ij} \bmod 2)$ and let
$2(L)=(l_{ij}'')$ with $l_{ij}''=0$ if $l_{ij}=0$,
$l_{ij}''=1$ if $l_{ij}$ is odd, and $l_{ij}''=2$ otherwise.
We call $2(L)$ the \emph{type} of~$L$.
Clearly, given two lists $L =(l_{ij})$ and $L' =(l'_{ij})$,
we have that $L\to L'$ if and only if $2(L)=2(L')$ and $l_{ij}\le l'_{ij}$
for each $i,j\in [n]$.

A feasible list $L_0$ is \emph{minimal} if there exists no feasible
list~$L^\star \ne L_0$ such that $L^\star \to L_0$. Thus a list~$L$ is feasible if and only
if there exists a minimal feasible list~$L_0$ of type $2(L)$ such that $L_0\to L$.

\paragraph{Our Contribution.}
As mentioned above, Yamanaka et
al.~\cite{yhuw-llr-CCCG18} showed that \listFeasibility in general is
NP-hard (which means that \tangleMinimization is
NP-hard as well). However, in their reduction, for some swaps the
number of occurrences is linear in the number of wires.
We strengthen their result by
showing that \listFeasibility is NP-hard even if all swaps
have constant multiplicity; see \cref{sec:complexity}.  Our
reduction uses a variant of \textsc{Not-All-Equal 3-SAT} (whereas
Yamanaka et al.\ used \textsc{1-in-3 3SAT}). Moreover we show that
for some types of lists, the problem is efficiently solvable;
see \cref{sec:fpt}.

For \tangleMinimization of simple lists for $n$ wires,
we present an
exact algorithm that is based on breadth-first search (BFS) in an
auxiliary graph and runs in $\oh(n!\varphi^n)$ time, where
$\varphi = (\sqrt{5}+1)/2 \approx 1.618$ is the golden ratio.
Recently, Baumann \cite{b-hmst-BTh20} has shown that the BFS can be
executed on a smaller auxiliary graph, which leads to a runtime of
$\oh(n!\psi^n)$ time, where
$\psi=(\sqrt[3]{9-\sqrt{69}}+\sqrt[3]{9+\sqrt{69}})/\sqrt[3]{18}
\approx 1.325$.
For general lists, we present an exact algorithm that is based on
dynamic programming and
runs in $\oh((2|L|/n^2+1)^{n^2/2}\varphi^n n\log|L|)$ time, where $L$ is the input list; see
\cref{sec:algorithms}.
Note that the runtime is polynomial in $|L|$ for fixed $n\ge 2$.

For \listFeasibility,
we use the dynamic programming algorithm from \cref{sec:algorithms}
and improve it to get an exponential-time algorithm with runtime
$\oh\big((2|L|/n^2+1)^{n^2/2}\cdot n^3 \log |L|\big)$;
see \cref{sec:fpt}.
The above runtimes are expressed in terms of the logarithmic cost model of
computation to show % explicitly
that the runtime depends only very
weakly on the length of~$L$ (actually on its largest entry).

Although we cannot characterize minimal feasible lists, we can bound their entries. Namely, we show
that, in a minimal feasible list of order~$n$, each swap occurs at
most $n^2/4+1$ times;
see Proposition~\ref{prop:minimal-feasible-lists-upper-bound}.
As a corollary, this yields that
\listFeasibility is in NP.
Combined with our exponential-time algorithm, this
also leads to an algorithm for testing
feasibility that is fixed-parameter tractable with respect to $n$.

\section{Complexity}
\label{sec:complexity}

Yamanaka et al.~\cite{yhuw-llr-CCCG18} showed that
\listFeasibility is NP-hard.  In their reduction, however,
some swaps have multiplicity $\Theta(n)$.  In this section, we show
that \listFeasibility is NP-hard even if all swaps have
multiplicity at most~8.  We reduce from \textsc{Positive NAE 3-SAT
	Diff}, a variant of \textsc{Not-All-Equal 3-SAT}.  Recall that in
\textsc{Not-All-Equal 3-SAT} one is given a Boolean formula in
conjunctive normal form with three literals per clause and the task is
to decide whether there exists a variable assignment such that in no
clause all three literals have the same truth value.  By Schaefer's
dichotomy theorem \cite{s-csp-STOC78}, \textsc{Not-All-Equal 3-SAT} is
NP-hard even if no negative (i.e., negated) literals are admitted.
In \textsc{Positive NAE 3-SAT Diff}, additionally each clause contains
three different variables.
We show that this variant is NP-hard, too.

\begin{lemma}
	\label{lem:PosNAEThreeSATDiff}
	\textsc{Positive NAE 3-SAT Diff} is NP-hard.
\end{lemma}

\begin{proof}
	We show NP-hardness of \textsc{Positive NAE 3-SAT Diff} by reduction
	from \textsc{Not-All-Equal 3-SAT}. Let $\Phi$
	be an instance of \textsc{Not-All-Equal 3-SAT} with variables
	$v_1, v_2, \dots, v_n$. First we show how to get rid of negative
	variables and then of multiple occurrences of the same variable in a
	clause.
	
	We create an instance $\Phi'$ of \textsc{Positive NAE 3-SAT Diff} as
	follows.  For every variable $v_i$, we introduce two new variables
	$x_i$ and $y_i$.  We replace each occurrence of $v_i$ by $x_i$ and
	each occurrence of $\neg v_i$ by $y_i$.  We need to force $y_i$ to
	be $\neg x_i$. To this end, we introduce the clause $(x_i \vee y_i
	\vee y_i)$.
	
	Now, we introduce three additional variables $a$, $b$, and $c$ that
	form the clause $(a \vee b \vee c)$.  Let $d = (x \vee x \vee y)$ be
	a clause that contains two occurrences of the same variable.  We
	replace $d$ by three clauses $(x \vee y \vee a)$,
	$(x \vee y \vee b)$, $(x \vee y \vee c)$. Since at least one of the
	variables $a$, $b$, or $c$ has to be true and at least one has to be
	false, $x$ and $y$ cannot have the same assignment, i.e.,
	$x = \neg y$.  Hence, $\Phi'$ is satisfiable if and only if $\Phi$
	is.  Clearly, the size of $\Phi'$ is polynomial in the size
	of~$\Phi$.
\end{proof}

\begin{theorem}
	\label{thm:hardness}
	\listFeasibility is NP-complete even if every pair of wires
	has at most eight swaps.
\end{theorem}

\begin{proof}
We split our proof, which uses gadgets for variables and clauses, into
several parts.  First, we introduce some
notation, then we give the intuition behind our reduction.  Next, we
present our variable and clause gadgets in more detail.  Finally, we show
the correctness of the reduction.

\proofparagraph{Notation.}
Recall that
we label the wires by their index in the initial permutation of a
tangle.  In particular, for a wire $\varepsilon$, its neighbor to the
right is wire $\varepsilon + 1$.  If a wire~$\mu$ is to the left of
some other wire~$\nu$ in the initial permutation,
then we write $\mu<\nu$.  If all wires in a set $M$
are to the left of all wires in a set~$N$ in the initial
permutation, then we write $M<N$.

\proofparagraph{Setup.}
Given an instance $\Phi = d_1 \wedge \dots \wedge d_m$ of
\textsc{Positive NAE 3-SAT Diff} with variables $w_1, \dots, w_n$, we
construct in polynomial time a list $L$ of swaps such that there is a
tangle $T$ realizing~$L$ if and only if $\Phi$ is a yes-instance.

In $L$, we have two inner wires $\lambda$ and $\lambda' = \lambda +1$
that swap eight times.  This yields two types of loops (see
\cref{fig:variable-gadget}): four $\lambda'$--$\lambda$ loops,
where $\lambda'$ is on the left and $\lambda$ is on the right side,
and three $\lambda$--$\lambda'$ loops with $\lambda$ on the left and
$\lambda'$ on the right side.  Notice that we consider only
\emph{closed} loops, which are bounded by swaps between $\lambda$ and
$\lambda'$.  In the following, we construct variable and clause
gadgets.  Each variable gadget will contain a specific wire that
represents the variable, and each clause gadget will contain a
specific wire that represents the clause.  The corresponding variable
and clause wires swap in one of the four $\lambda'$--$\lambda$ loops.
We call the first two $\lambda'$--$\lambda$ loops \emph{true-loops},
and the last two $\lambda'$--$\lambda$ loops \emph{false-loops}.  If
the corresponding variable is true, then the variable wire swaps with
the corresponding clause wires in a true-loop, otherwise in a
false-loop.

Apart from $\lambda$ and $\lambda'$, our list $L$ contains (many)
other wires, which we split into groups.  For every $i \in [n]$, we
introduce sets $V_i$ and $V'_i$ of wires that together form the gadget
for variable $w_i$ of~$\Phi$.  These sets are ordered (initially)
$V_n < V_{n-1} < \dots < V_1 < \lambda < \lambda' < V'_1 < V'_2 <
\dots < V'_n$; the order of the wires inside these sets will be
detailed in the next two paragraphs.  Let
$V = V_1 \cup V_2 \cup \dots\cup V_n$ and
$V' = V'_1 \cup V'_2 \cup \dots \cup V'_n$.  Similarly, for every
$j \in [m]$, we introduce a set $C_j$ of wires that contains a
\emph{clause wire} $c_j$ and three sets of wires $D^1_j$, $D^2_j$, and
$D^3_j$ that represent occurrences of variables in a clause $d_j$ of
$\Phi$.  The wires in~$C_j$ are ordered $D^3_j < D^2_j < D^1_j < c_j$.
Together, the wires in $C = C_1 \cup C_2 \cup \dots \cup C_m$
represent the clause gadgets; they are ordered
$V < C_m < C_{m-1} < \dots < C_1 < \lambda$.  Additionally, our list
$L$ contains a set~$E=\{\varphi_1,\dots,\varphi_7\}$ of wires that
will make our construction rigid enough.  The order of all wires in
$L$ is $V < C < \lambda < \lambda' < E < V'$.  Now we present our
gadgets in more detail.

\proofparagraph{Variable gadget.}
First we describe the variable gadget,
which is illustrated in \cref{fig:variable-gadget}.
For each variable $w_i$ of $\Phi$, $i \in [n]$, we introduce two sets of
wires $V_i$ and~$V'_i$.  Each $V'_i$ contains a \emph{variable wire}
$v_i$ that has four swaps with $\lambda$ and no swaps with
$\lambda'$.  Therefore, $v_i$ intersects at least one and at most two
$\lambda'$--$\lambda$ loops. In order to prevent~$v_i$ from
intersecting both a true- and a false-loop, we introduce two wires
$\alpha_i \in V_i$ and $\alpha'_i \in V'_i$ with
$\alpha_i < \lambda < \lambda' < \alpha'_i < v_i$; see
\cref{fig:variable-gadget}.  These wires neither swap with $v_i$
nor with each other, but they have two swaps with both $\lambda$ and
$\lambda'$.  We want to force $\alpha_i$ and $\alpha'_i$ to have the
two true-loops on their right and the two false-loops on their left,
or vice versa.  This will ensure that $v_i$ cannot reach both a true-
and a false-loop simultaneously.

To this end, we introduce, for $j \in [5]$, a \emph{$\beta_i$-wire}
$\beta_{i,j} \in V_i$ and a \emph{$\beta'_i$-wire}
$\beta'_{i,j} \in V'_i$.  These are ordered
$\beta_{i,5} < \beta_{i,4} < \dots < \beta_{i,1} < \alpha_i$ and
$\alpha'_i < \beta'_{i,1} < \beta'_{i,2} < \dots < \beta'_{i,5} <
v_i$.  Every pair of $\beta_i$-wires as well as every pair of
$\beta'_i$-wires swaps exactly once. Neither $\beta_i$- nor
$\beta'_i$-wires swap with $\alpha_i$ or $\alpha'_i$. Each
$\beta'_i$-wire has four swaps with~$v_i$. Moreover,
$\beta_{i,1}, \beta_{i,3}, \beta_{i,5}, \beta'_{i,2}, \beta'_{i,4}$
swap with $\lambda$ twice.  Symmetrically,
$\beta_{i,2}, \beta_{i,4}, \beta'_{i,1}, \beta'_{i,3}, \beta'_{i,5}$
swap with $\lambda'$ twice; see \cref{fig:variable-gadget}.

We use the $\beta_i$- and $\beta'_i$-wires to fix the minimum number
of $\lambda'$--$\lambda$ loops that are on the left of $\alpha_i$ and
on the right of $\alpha'_i$, respectively.  Note that, together with
$\lambda$ and $\lambda'$, the $\beta_i$- and $\beta'_i$-wires have the
same rigid structure as the wires
shown in \cref{fig:many_loops}.

\begin{figure}[tb]
	\centering
	\includegraphics[width=\linewidth]{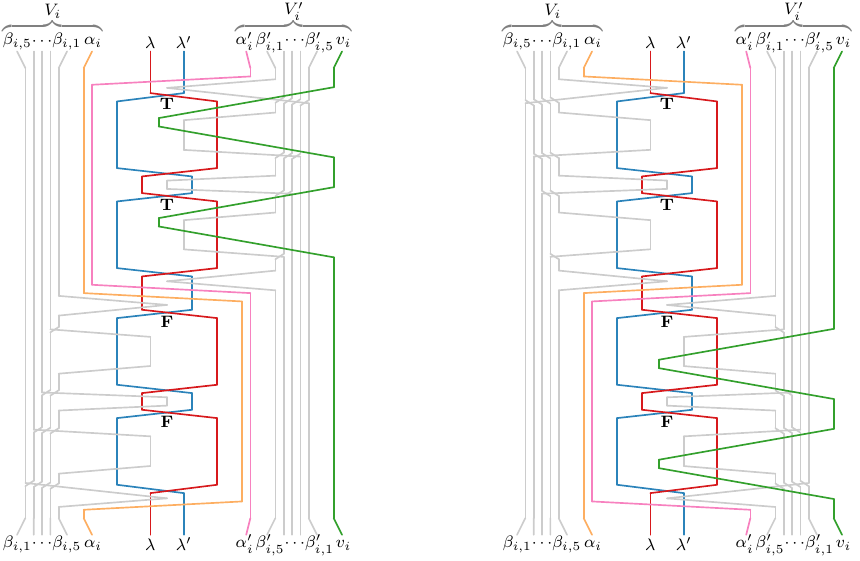}
	\caption{A variable gadget with a variable wire $v_i$ that
		corresponds to the variable~$w_i$ that is true (left)
                or false (right).
		The $\lambda$--$\lambda'$ loops are labeled
		\textbf{T} for true and \textbf{F} for false.}
	\label{fig:variable-gadget}
\end{figure}

This means that there is a unique order of swaps between the
$\beta_i$-wires and $\lambda$ or $\lambda'$, i.e., for $j \in [4]$,
every pair of $(\beta_{i,j+1},\lambda)$ swaps (or
$(\beta_{i,j+1},\lambda')$ swaps, depending on the parity of $j$) can
be done only after the pair of $(\beta_{i,j},\lambda')$ swaps (or
$(\beta_{i,j},\lambda)$ swaps, respectively).  We have the same rigid
structure on the right side with $\beta'_i$-wires.  Hence, there are
at least two $\lambda'$--$\lambda$ loops to the left of $\alpha_i$ and
at least two to the right of $\alpha'_i$.  Since $\alpha_i$ and
$\alpha'_i$ do not swap, there cannot be a $\lambda'$--$\lambda$ loop
that appears simultaneously on both sides.

Note that the $(\lambda,\lambda')$ swaps that belong to the same side
need to be consecutive, otherwise $\alpha_i$ or $\alpha'_i$ would have
to swap more than twice with $\lambda$ and $\lambda'$.  Thus, there are
only two ways to order the swaps among the wires $\alpha_i$,
$\alpha'_i$, $\lambda$, $\lambda'$; the order is either
$(\alpha'_i,\lambda')$, $(\alpha'_i,\lambda)$, four times
$(\lambda,\lambda')$, $(\alpha'_i,\lambda)$,
$(\alpha'_i,\lambda')$, $(\alpha_i,\lambda)$,
$(\alpha_i,\lambda')$, four times $(\lambda,\lambda')$,
$(\alpha_i,\lambda')$, $(\alpha_i,\lambda)$ (see
\cref{fig:variable-gadget} (left)) or the reverse (see
\cref{fig:variable-gadget} (right)).  It is easy to see that in the
first case $v_i$ can reach only the first two $\lambda'$--$\lambda$
loops (the true-loops), and in the second case only the last two (the
false-loops).

To avoid that the gadget for variable~$w_i$ restricts the proper
functioning of the gadget for some variable~$w_j$ with $j>i$, we add
the following swaps to~$L$: for any $j>i$, $\alpha_j$ and $\alpha'_j$
swap with both $V_i$ and $V'_i$ twice, the $\beta_j$-wires swap with
$\alpha'_i$ and $V_i$ four times, and, symmetrically,
the $\beta'_j$-wires swap with $\alpha_i$ and $V'_i$ four times,
and~$v_j$ swaps with $\alpha_i$ and~$V'_i$ six times.

We briefly explain these multiplicities.
Wire $\alpha_j$ ($\alpha'_j$, resp.)
swaps with all wires in~$V_i$ and~$V'_i$ twice
so that it reaches the corresponding
$\lambda$--$\lambda'$ or $\lambda'$--$\lambda$ loops
by first crossing all wires of $V_i$ ($V'_i$).
Then $\alpha_j$ ($\alpha'_j$) crosses all wires of $V'_i$ ($V_i$)
and finally goes back crossing all of them a second time.
This way, there are no restrictions for $\alpha_j$ ($\alpha'_j$)
regarding which $\lambda'$--$\lambda$ loops it encloses.
For the variable wire~$v_j$, see \cref{fig:two-variables}.
It swaps six times with~$\alpha_i$
and $V'_i$, which guarantees that $v_j$
can reach both upper or both lower $\lambda'$--$\lambda$ loops.
If $w_i$ and $w_j$ represent the same truth value,
$v_j$ needs a total of four swaps with each of the
$\beta'_i$-wires to enter a loop and to go back.
In this case, $v_j$ can make the six swaps with $\alpha_i$
and $\alpha'_i$ next to the
$\lambda'$--$\lambda$ loops representing the other truth value.
Then, $v_j$ uses the two extra swaps with each of the
$\beta'_i$-wires to reach $\alpha_i$ and $\alpha'_i$;
see \cref{fig:two-variables} (left).
If $w_i$ and $w_j$ represent distinct truth values,
$v_j$ needs a total of four swaps with each of the $\beta'_i$-wires
and with $\alpha'_i$ and $\alpha_i$ to enter a loop and to go back.
We can accommodate the two extra swaps with each of these wires
afterwards; see \cref{fig:two-variables} (right).
For the $\beta_j$-wires (and $\beta'_j$-wires),
we use the same argument as for~$v_j$ above,
but each of these wires has only four swaps
with each of the other wires in $V_i \cup \{\alpha'_i\}$
(or $V'_i \cup \{\alpha_i\}$)
because a $\beta$-wire (or a $\beta'$-wire) needs to reach only one
$\lambda$--$\lambda'$ or $\lambda'$--$\lambda$ loop
instead of two (as~$v_j$).

\begin{figure}[tb]
	\centering
	\includegraphics[width=\linewidth]{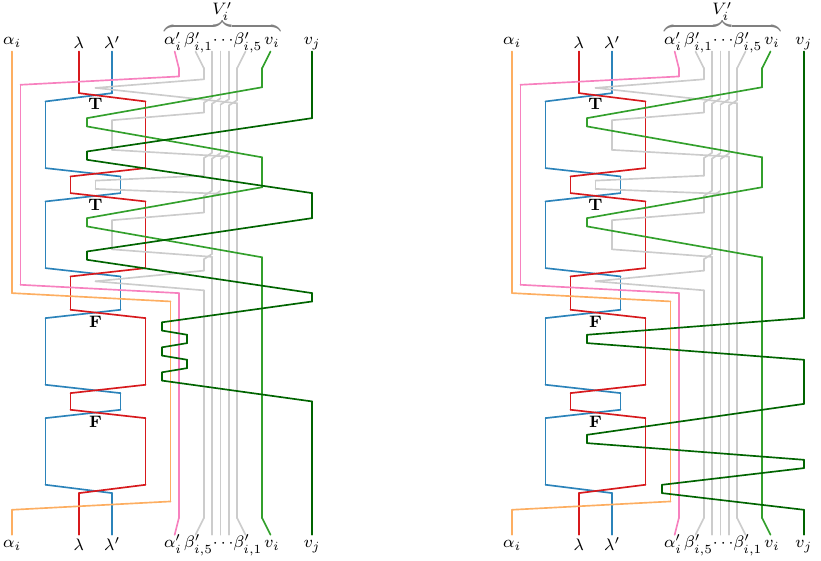}
	\caption{A realization of swaps between the variable wire $v_j$ and
		all wires that belong to the variable gadget corresponding to the
		variable $w_i$.  On the left the variables $w_i$ and $w_j$ are both
		true, and on the right $w_i$ is true, whereas $w_j$ is false.}
	\label{fig:two-variables}
\end{figure}

\proofparagraph{Clause gadget.}
For every clause $d_j$ from $\Phi$, $j \in [m]$, we introduce a set of
wires~$C_j$.  It contains the clause wire $c_j$ that has eight swaps
with $\lambda'$.  We want to force each $c_j$ to appear in all
$\lambda'$--$\lambda$ loops.  To this end, we use
(once for all clause gadgets) the set~$E$ with the
seven \emph{$\varphi$-wires} $\varphi_1, \dots, \varphi_7$ ordered
$\varphi_1 < \dots < \varphi_7$.  They create a rigid structure
according to Observation~\ref{obs:unique} similar to the $\beta_i$-wires.
Each pair of $\varphi$-wires swaps exactly once.
For each $k \in [7]$, if $k$ is
odd, then $\varphi_k$ swaps twice with $\lambda$ and twice with $c_j$ for
every $j \in [m]$.  If $k$ is even, then $\varphi_k$ swaps twice with
$\lambda'$. Since $c_j$ does not swap with~$\lambda$, each pair of
swaps between~$c_j$ and a $\varphi$-wire with odd index appears inside
a $\lambda'$--$\lambda$ loop.  Due to the rigid structure, each of
these pairs of swaps occurs in a different $\lambda'$--$\lambda$ loop;
see \cref{fig:clause-gadget}.

\begin{figure}[tb]
	\centering
	\includegraphics[width=0.47\linewidth]{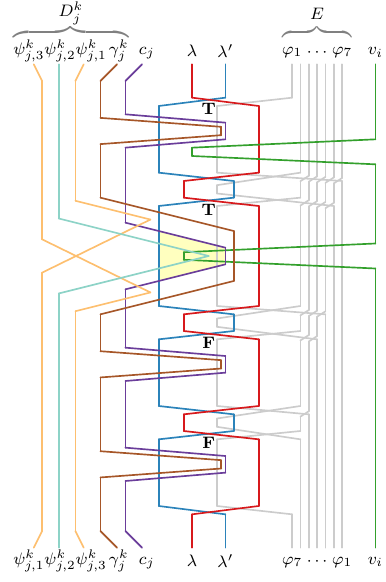}
	\caption{A gadget for clause $c_j$ showing only one of the
          three variable wires, namely $v_i$.  The region shaded in
          yellow is the arm of $c_j$ that is protected from other
          variables by $\gamma_j^k$.}
	\label{fig:clause-gadget}
\end{figure}

If a variable $w_i$ belongs to a clause $d_j$, then our list
$L$ contains two $(v_i,c_j)$ swaps.  Since every clause has exactly
three different positive variables, we want to force variable wires
that belong to the same clause to swap with the corresponding clause
wire in different $\lambda'$--$\lambda$ loops.  This way, every clause
contains at least one true and at least one false variable if $\Phi$ is
satisfiable.

\begin{figure}[t]
	\centering
	\includegraphics[width=\textwidth]{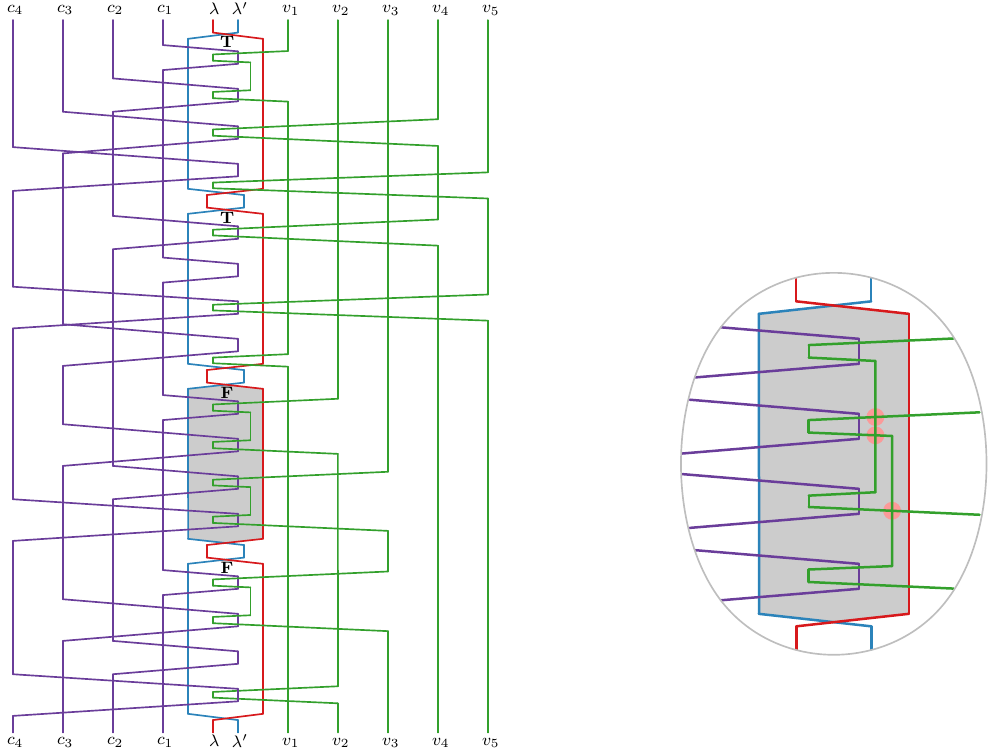}
	\caption{Tangle obtained from the satisfiable formula
		$\Phi = (w_1 \vee w_2 \vee w_3) \wedge (w_1 \vee w_3 \vee w_4) \wedge
		(w_2 \vee w_3 \vee w_4) \wedge (w_2 \vee w_3 \vee w_5)$.  Here,
		$w_1$, $w_4$ and $w_5$ are set to true, whereas $w_2$ and $w_3$
		are set to false.  We show only %
		$\lambda$, $\lambda'$, and all variable and clause wires. \\
		Inset:
		problems that occur if variable wires swap with clause wires in a
		different order.}
	\label{fig:formula}
\end{figure}

We call the part of a clause wire $c_j$ that is inside a
$\lambda'$--$\lambda$ loop %
an \emph{arm} of $c_j$.  We want to ``\emph{protect}'' the arm
that is intersected by a variable wire from other variable wires.  To
this end, for every occurrence $k \in [3]$ of a variable in~$d_j$, we
introduce four more wires.  The wire $\gamma_j^k$ will protect the arm
of $c_j$ that the variable wire of the $k$-th variable of $d_j$
intersects.
Below we detail how to realize this protection.  For now, just note
that, in order not to restrict the
choice of the $\lambda'$--$\lambda$ loop, $\gamma_j^k$ swaps twice
with $\varphi_\ell$ for every odd $\ell \in [7]$.  Similarly to $c_j$,
the wire~$\gamma_j^k$ has eight swaps with $\lambda'$ and appears once
in every $\lambda'$--$\lambda$ loop.  Additionally, $\gamma_j^k$ has
two swaps with~$c_j$.

We force $\gamma_j^k$ to protect the correct arm in the following
way.  Consider the $\lambda'$--$\lambda$ loop where an arm of $c_j$
swaps with a variable wire $v_i$.  We want the order of swaps along
$\lambda'$ inside this loop to be fixed as follows: $\lambda'$ first
swaps with $\gamma_j^k$, then twice with~$c_j$, and then again
with~$\gamma_j^k$.  This would prevent all variable wires that do not
swap with $\gamma_j^k$ from reaching the arm of~$c_j$.  To achieve
this, we introduce three \emph{$\psi_j^k$-wires}
$\psi_{j,1}^k, \psi_{j,2}^k, \psi_{j,3}^k$ with
$\psi_{j,3}^k < \psi_{j,2}^k < \psi_{j,1}^k < \gamma_j^k$.

The $\psi_j^k$-wires also have the rigid structure
according to Observation~\ref{obs:unique} similar to the
$\varphi$- and $\beta_i$-wires, so that there is a unique order of swaps
along each $\psi_j^k$-wire.  Each pair of $\psi_j^k$-wires swaps
exactly once, $\psi_{j,1}^k$ and $\psi_{j,3}^k$ have two swaps with
$c_j$, and $\psi_{j,2}^k$ has two swaps with $\lambda'$ and
$v_i$. Note that no $\psi_j^k$-wire swaps with~$\gamma_j^k$.  Also,
since $\psi_{j,2}^k$ does not swap with $c_j$, the
$(\psi_{j,2}^k,v_i)$ swaps can appear only inside the
$\lambda'$--$c_j$ loop that contains the arm of~$c_j$ we want to
protect from other variable wires.  Since $c_j$ has to swap with
$\psi_{j,1}^k$ before and with $\psi_{j,3}^k$ after the
$(\psi_{j,2}^k,\lambda')$ swaps, and since there are only two swaps
between $\gamma_j^k$ and $c_j$, there is no way for any variable wire
except for $v_i$ to reach the arm of $c_j$ without also intersecting
$\gamma_j^k$; see \cref{fig:clause-gadget}.

Finally, we consider the behavior
among $D^1_j$, $D^2_j$, and $D^3_j$,
as well as the behavior of the wires from $C_j$
with other clause and variable gadgets.
For $k, \ell \in [3]$ with $k \ne \ell$,
every $\psi_j^k$-wire has two swaps with every $\psi_j^\ell$-wire
in order to not restrict the $\psi$-wires to a specific
$\lambda'$--$\lambda$ loop.
Furthermore, $\gamma_j^k$ has four swaps with every $\psi_j^\ell$-wire
and two swaps with every wire~$\gamma_j^\ell$.
Note that, if $k > \ell$, it suffices for~$\gamma_j^k$
to cross each $\gamma_j^\ell$ at the $\lambda'$--$\lambda$ loop
where $\gamma_j^k$ ``protects'' $c_j$ and to go back directly afterwards.
Also, $\gamma_j^k$ uses one swap to cross
the $\psi_j^\ell$-wires before the first $\lambda'$--$\lambda$ loop,
then, using two swaps, the $\psi_j^\ell$-wires
cross $\gamma_j^k$ next to the $\lambda'$--$\lambda$ loop
where $\gamma_j^\ell$ ``protects'' $c_j$,
and $\gamma_j^k$ uses the fourth swap to go back after
the last $\lambda'$--$\lambda$ loop.

For every $i < j$ and every $k \in [3]$, the wires
$c_j$ and $\gamma_j^k$ have eight swaps with every wire in~$C_i$,
which allows $c_j$ and $\gamma_j^k$ to enter
every $\lambda'$--$\lambda$ loop and to go back.
Similarly, every $\psi_j^k$-wire
has two swaps with every wire in~$C_i$,
which allows the $\psi_j^k$-wire to reach
one of the $\lambda'$--$\lambda$ loops.
Since all wires in~$V$ are to the left of all wires in~$C$,
each wire in~$C$ swaps twice with all wires in~$V$
(including the $\alpha$-wires) and twice with all $\alpha'$-wires.
Finally, the $\varphi$-wires with odd index
have (beside two swaps with every $c_j$)
two swaps with every $\gamma_j^k$,
and all $\varphi$-wires have, for $i \in [n]$, two swaps
with every wire in $(V'_i \setminus \{v_i\}) \cup \{\alpha_i\}$
and four swaps with every wire $v_i$
to let the wires of the variable gadgets enter (or cross) the
$\lambda$--$\lambda'$ or $\lambda'$--$\lambda$ loops.

Note that the order of the arms of the clause wires inside a
$\lambda'$--$\lambda$ loop cannot be chosen arbitrarily.  If a
variable wire intersects more than one clause wire, the arms of these
clause wires occur consecutively, as for~$v_2$ and~$v_3$ in the
shaded region in \cref{fig:formula}.  If we had an interleaving
pattern of %
variable wires (see inset), say $v_2$ first intersects~$c_1$, then
$v_3$ intersects~$c_2$, then $v_2$ intersects~$c_3$, and finally $v_3$
intersects~$c_4$, then $v_2$ and~$v_3$ would need more swaps
than calculated in the analysis above.

\proofparagraph{Correctness.}
Clearly, if $\Phi$ is satisfiable, then there is a tangle obtained from
$\Phi$ as described above that realizes the list $L$, so $L$ is feasible;
see \cref{fig:formula} for an example.
On the other hand, if there is a tangle that realizes the list~$L$
that we obtain from the reduction, then $\Phi$ is satisfiable.  This
follows from the rigid structure of a tangle that realizes~$L$.  The
only flexibility is in which type of loop (true or false) a variable
wire swaps with the corresponding clause wire.  As described above, a
tangle exists if, for each clause, the corresponding three variable
wires swap with the clause wire in three different loops (at least one
of which is a true-loop and at least one is a false-loop).  In this
case, the position of the variable wires yields a truth assignment
satisfying~$\Phi$.

\proofparagraph{Membership in NP.}
To show that \listFeasibility is in NP, we proceed as
indicated in the introduction.  Given a list $L=(l_{ij})$, we guess a
list $L'=(l'_{ij})$ with $2(L)=2(L')$ and
$l'_{ij} \le \min \{ l_{ij}, n^2/4+1 \}$ together with a permutation
of its $\oh(n^4)$ swaps.  Then we can efficiently test whether we can
apply the swaps in this order to~$\id_n$.
If yes, then the list~$L'$ is
feasible (and, due to $L' \to L$, a witness for the feasibility
of~$L$), otherwise we discard it. Note that $L$ is
feasible if and only if such a list $L'$ exists and is feasible.
One direction is obvious by the definition of minimal feasible list.
To show the other direction, we assume that $L$ is feasible. Thus
there exists a minimal feasible list $L' = (l_{ij})$ of type $2(L)$
and by Proposition~\ref{prop:minimal-feasible-lists-upper-bound}
$l'_{ij} \le n^2/4+1$.

Note that this result does not show that the decision version of
\tangleMinimization is also in NP
because the minimum height can be exponential in the size
of the input.
\end{proof}

\section{Algorithms for Minimizing Tangle Height}
\label{sec:algorithms}

The two algorithms that we describe in this section test whether a
given list is feasible and, if yes, construct a height-optimal tangle
realizing the list.  We start with an observation and some definitions.

\begin{lemma}
	\label{lem:adjacent}
	Given a permutation $\pi \in S_n$, the number of permutations adjacent to $\pi$
	is $F_{n+1}-1$, where $F_n$ is the $n$-th Fibonacci number.
\end{lemma}

\begin{proof}
	By induction on~$n$, we prove the slightly different claim that
	the set~$P(\pi)$ consisting of $\pi$ and its adjacent permutations
	has size $F_{n+1}$.
	
	For $n = 1$, there is only one permutation `1'
	and, hence, $|P(1)| = 1$.
	Furthermore, for $n = 2$, there are only two permutations
	`12` and `21', which are also adjacent.
	Thus, $|P(12)| = |P(21)| = 2$.
	These are the Fibonacci numbers $F_2$ and $F_3$.
	
	Let $n > 2$ and $\pi \in S_n$.
	Note that we can partition the permutations in $P(\pi)$ into two groups.
	The first group $P_1(\pi)$ contains the permutations of~$P(\pi)$
	where the last wire is the same as in~$\pi$,
	and the second group $P_2(\pi)$ contains the permutations of~$P(\pi)$
	where the last wire is different from~$\pi$.
	Clearly, $|P(\pi)| = |P_1(\pi)| + |P_2(\pi)|$.
	
	To obtain $P_1(\pi)$, we first remove the last wire of~$\pi$.
	This yields a permutation $\pi' \in S_{n-1}$ (after possibly
        renaming the wires to avoid a gap in our naming scheme).
	By our inductive hypothesis, $|P(\pi')| = F_{n}$.
	Then, we append the last wire of $\pi$ to every
	permutation in $P(\pi')$.  This yields~$P_1(\pi)$.
	
	For $P_2(\pi)$, observe that the last two wires in every
        permutation in $P_2(\pi)$ are swapped compared to~$\pi$.
	We remove these two wires from~$\pi$.  This yields a
        permutation~$\pi'' \in S_{n-2}$.
	Again by our inductive hypothesis, $|P(\pi'')| = F_{n-1}$.
	Then, we append the last two wires of~$\pi$ in swapped order to
	every permutation in $P(\pi'')$.  This yields~$P_2(\pi)$.

	Summing up, we get $|P(\pi)| = F_n + F_{n-1} = F_{n+1}$.
\end{proof}

We say that a list is consistent if the final
positions of all wires form a permutation of~$S_n$.
In \cref{sec:fpt}, we show that this condition is sufficient for the
feasibility of simple (and, more generally, for odd) lists.
Clearly, an even list is always consistent.
More formally, we define consistency as follows.
For a permutation $\pi\in S_n$ and a list $L=(l_{ij})$, we define
the map $\pi L\colon [n]\to[n]$ by
\[i \;\mapsto\; \pi(i)+|\{j\colon  \pi(i)<\pi(j)\le n\text{ and
  }l_{ij}\text{ odd}\}| - |\{j\colon 1\le \pi(j)<\pi(i)\text{ and
  }l_{ij}\text{ odd}\}|.\]
For each wire $i \in [n]$, $\pi L(i)$ is the final position of $i$,
that is, the position after all swaps in~$L$ have been applied to $\pi$.
A list~$L$ is called \emph{$\pi$-consistent} if $\pi L\in S_n$, or,
more rigorously, if
$\pi L$ induces a permutation of $[n]$. An $\id_n$-consistent
list is \emph{consistent}.  For example, the list $\{(1,2), (2,3), (1,3)\}$ is
consistent, whereas the list $\{(1,3)\}$ is not.
If~$L$ is not consistent, then it is clearly not feasible.
However, not all consistent lists are feasible.
For example, the list $\{13, 13\}$ is consistent but not feasible.
For a list $L=(l_{ij})$, we define $1(L)=(l_{ij} \bmod 2)$.  Since
$\id_n L=\id_n 1(L)$, the list $L$ is consistent if and only if $1(L)$
is consistent.  We can compute $1(L)$ and check its consistency in
$\oh(n+|1(L)|) \subseteq \oh(n^2)$ time
(see Proposition~\ref{prop:feasibility-simple} for details).
Hence, in the sequel we assume that all lists are consistent.

For any permutation $\pi\in S_n$, we define the simple list
$L(\pi)=(l_{ij})$ such that for $1 \le i<j \le n$, $l_{ij}=0$ if $\pi(i)<\pi(j)$, and $l_{ij}=1$ otherwise.

We need the following two lemmas for our algorithms.

\begin{lemma}
	\label{lem:Lpi1}
	For every permutation $\pi\in S_n$,
	$L(\pi)$ is the unique simple list
	with $\id_n L(\pi)=\pi$. In other words,
	$L(\pi)$ is the unique simple list
	such that applying all swaps
	in $L(\pi)$ to $id_n$ yields $\pi$.
\end{lemma}

\begin{proof}
	By definition, $\id_n L(\pi)$ is a map from $[n]$ to $\mathbb Z$,
        \begin{align*}
		i \mapsto& \; i+\left|\{j\colon i<j\le n\mbox{ and }\pi(i)>\pi(j)\}\right|
		-\left|\{j\colon 1\le j<i\mbox{ and }\pi(i)<\pi(j)\}\right|\\
		= & \; i+\left|\{j\colon i<j\le n\mbox{ and }\pi(i)>\pi(j)\}\right| + \left|\{j\colon  1\le j<i\mbox{ and }\pi(i)>\pi(j)\}\right|\\
		& \; \phantom{i} - \left|\{j\colon 1\le j<i\mbox{ and }\pi(i)>\pi(j)\}\right| - \left|\{j\colon 1\le j<i\mbox{ and }\pi(i)<\pi(j)\}\right|\\
		= & \; i+\left|\{j\colon 1<j\le n\mbox{ and }\pi(i)>\pi(j)\}\right| - \left|\{j\colon 1\le j<i\}\right|\\
		= & \; i+(\pi(i)-1)-(i-1)=\pi(i).
         \end{align*}
	
	Assume that $L=(l_{ij})$ is a simple list such that $\id_n L=\pi$. That is, for each $i\in [n]$, we
	have
        \[\pi(i)=i+|\{j\colon i<j\le n \mbox{ and }l_{ij}=1\}|-|\{j\colon 1\le j<i\mbox{ and }l_{ij}=1\}|.\]
	
	We show that the list $L$ is uniquely determined by the permutation $\pi$ by induction with respect
	to $n$. For $n=2$, there exist only two simple lists
	$\left(\begin{smallmatrix}0&0\\0&0\end{smallmatrix}\right)$ and
	$\left(\begin{smallmatrix}0&1\\1&0\end{smallmatrix}\right)$. Since
	$\id_2\left(\begin{smallmatrix}0&0\\0&0\end{smallmatrix}\right)=\id_2$
	and
	$\id_2\left(\begin{smallmatrix}0&1\\1&0\end{smallmatrix}\right)=21$,
	we have uniqueness.  Now assume that $n\ge 3$.
        Then, for $k=\pi^{-1}(n)$, we have
        \begin{align*}
		n=\pi(\pi^{-1}(n))
		=\pi(k)
		=k+|\{j\colon k<j\le n\mbox{ and }l_{kj}=1\}|
		-|\{j\colon 1\le j<k\mbox{ and }l_{kj}=1\}|.
        \end{align*}
	Since
        \[|\{j\colon k<j\le n\mbox{ and }l_{kj}=1\}|\le |\{j\colon k<j\le n\}|=n-k\]
	and
        \[|\{j\colon 1\le j<k\mbox{ and }l_{kj}=1\}|\ge 0,\]
	the equality holds if and only if $l_{kj}=1$ for each $k<j\le n$ and $l_{kj}=0$ for
	each $1\le j<k$. These conditions determine the $k$-th row (and
	column) of the matrix~$L$.
	
	It is easy to see that a map $\pi'\colon [n-1]\to\mathbb Z$ such that $\pi'(i)=\pi(i)$ for $i<k$ and
	$\pi'(i)=\pi(i+1)$ for $i\ge k$ is a permutation.
	
	Let $L'=(l'_{ij})$ be a simple list of order $n-1$ obtained from $L$ by removing the $k$-th row and
	column. For each $i\in [n-1]$, we have
        \[\id_n L'(i)=i+|\{j\colon i<j\le n-1 \mbox{ and }l'_{ij}
		=1\}|-|\{j\colon 1\le j<i\mbox{ and }l'_{ij}=1\}|.\]
	If $i<k$ then $l_{ik}=0$.  Thus,
        \begin{align*}
		\id_n L'(i)&=i+|\{j\colon i<j\le n \mbox{ and }l_{ij}=1\}|
		-|\{j\colon 1\le j<i\mbox{ and }l_{ij}=1\}|\\
		&=\pi(i)=\pi'(i).
        \end{align*}
	If on the other hand $i\ge k$, then $l_{i+1,k}=1$.  Hence,
        \begin{align*}
            \id_n L'(i) =& \, i+|\{j\colon i<j\le n-1 \mbox{ and }l_{i+1,j+1}=1\}|\\
                        & \, \phantom{i} -|\{j\colon 1\le j<i+1\mbox{ and }l_{i+1,j}=1\}|+1\\
                        =& \, i+1+|\{j\colon i+1<j+1\le n \mbox{ and }l_{i+1,j+1}=1\}|\\
                        & \, \phantom{i} -|\{j\colon 1\le j<i+1\mbox{ and }l_{i+1,j}=1\}|\\
                        =& \, \pi(i+1)=\pi'(i).
        \end{align*}
	
	Thus, $\id_n L'=\pi'$. By the inductive hypothesis, the list $L'$ is uniquely determined by the
	permutation $\pi'$, so the list $L$ is uniquely determined by the permutation $\pi$.
\end{proof}

Given a tangle $T$, we define the list \emph{$L(T)$} $=(l_{ij})$,
where $l_{ij}$ is the number of occurrences of swap~$(i,j)$ in~$T$.

\begin{lemma}
	\label{lem:final-independence}
	For every tangle $T=\langle\pi_1, \pi_2,\dots, \pi_h\rangle$,
	we have $\pi_1L(T)=\pi_h$.
\end{lemma}

\begin{proof}
	We have $L(T)=(l_{ij})$, where
          \[l_{ij}=|\{t \colon 1\le t<h, \pi^{-1}_t\pi_{t+1}(i)=j
            \text{ and }\pi^{-1}_t\pi_{t+1}(j)=i\}|\]
	for each distinct $i,j\in [n]$. If $\pi_1(i)<\pi_1(j)$, then it is easy to see
	that $\pi_h(i)<\pi_h(j)$ if and only if $l_{ij}$ is even.
	On the other hand, by definition, for each $i\in [n]$, we have that
        \begin{align*}
		\pi_1 L(T)(i) =& \; \pi_1(i)+|\{j\colon \pi_1(i)<\pi_1(j)\le n\mbox{ and }l_{ij}\mbox{ is odd}\}|\\
		& \; \phantom{\pi_1(i)} -|\{j\colon 1\le \pi_1(j)<\pi_1(i)\mbox{ and }l_{ij}\mbox{ is odd}\}|\\
		=& \; \pi_1(i)+|\{j\colon \pi_1(i)<\pi_1(j)\le n\mbox{ and }\pi_h(i)>\pi_h(j)\}|\\
		& \; \phantom{\pi_1(i)} -|\{j\colon 1\le \pi_1(j)<\pi_1(i)\mbox{ and } \pi_h(i)<\pi_h(j)\}|.
        \end{align*}
	Now, similarly to the beginning of the proof of \cref{lem:Lpi1}, we can
	show that $\pi_1 L(T)(i)=\pi_h(i)$.
\end{proof}

\paragraph{Simple lists.}

We first solve \tangleMinimization for simple lists.
Our algorithm does BFS in an appropriately defined auxiliary graph.

\begin{theorem}
	\label{thm:exact-simple}
	For a simple list of order~$n$, \tangleMinimization
	can be solved in $\oh(n!\varphi^n)$ time,
	where $\varphi = (\sqrt{5}+1)/2 \approx 1.618$ is the golden ratio.
\end{theorem}

\begin{proof}
Let $L$ be a consistent simple list.
Wang's algorithm~\cite{w-nrsic-DAC91} creates a \emph{simple} tangle
from $\id_n L$, i.e., a tangle where all its permutations are distinct.
Thus~$L$ is feasible. Note that the height of a simple
tangle is at most~$n!$. Let
$T=(\id_n{=}\pi_1,\pi_2,\ldots,\pi_h{=}\id_n L)$ be any tangle
such that $L(T)$ is simple.
Then, by \cref{lem:final-independence}, $\id_n L(T)=\pi_h$.
By \cref{lem:Lpi1}, $L(\pi_h)$ is the unique simple list with $\id_n L(\pi_h)=\pi_h=\id_n L$,
so $L(T)=L(\pi_h)=L$ and thus~$T$ is a realization of~$L$.

We compute an optimal tangle realizing $L=(l_{ij})$ as follows.
Consider the graph $G_L$ whose vertex set $V(G_L)$ consists of all
permutations $\pi\in S_n$ with $L(\pi) \le L$ (componentwise).
A directed edge $(\pi,\sigma)$ between vertices $\pi,\sigma\in V(G_L)$
exists if and only if~$\pi$ and~$\sigma$ are adjacent permutations
and $L(\pi)\cap L(\pi^{-1}\sigma)=\varnothing$; the latter means that
the set of (disjoint) swaps that transforms $\pi$ to $\sigma$ cannot contain
swaps from the set that transforms $\id_n$ to $\pi$.
The graph $G_L$ has at most $n!$ vertices and maximum degree
$F_{n+1}-1$; see \cref{lem:adjacent}.
Note that $F_n=(\varphi^{n}-(-\varphi)^{-n})/\sqrt{5} \in \Theta(\varphi^n)$.
Furthermore, for each $h\ge 0$, there is a natural
bijection between tangles of height $h+1$ realizing $L$ and paths of length $h$
in the graph $G_L$ from the initial permutation $\id_n$ to the permutation $\id_n L$.
A~shortest such path can be found by BFS in
$\oh(E(G_L))=\oh(n!\varphi^n)$ time.
\end{proof}

\paragraph{General lists.}

Now we solve \tangleMinimization for general lists.
We employ a dynamic program (DP).

\begin{theorem}
	\label{thm:exact-general}
	For a list~$L$ of order~$n$, \tangleMinimization can
	be solved in $\oh((2|L|/n^2+1)^{n^2/2} \cdot \varphi^n \cdot n \cdot \log|L|)$ time,
	where $\varphi = (\sqrt{5}+1)/2 \approx 1.618$ is the golden ratio.
\end{theorem}

\begin{proof}
Without loss of generality, assume that $|L|\ge n/2$;
otherwise, there is a wire $k \in [n]$ that does not belong to any swap.
This wire splits $L$ into smaller lists with independent realizations.
(If there is a swap $(i,j)$ with $i<k<j$, then $L$ is infeasible.)

Let $L=(l_{ij})$ be the given list.  We describe a DP that computes
the height of an optimal tangle realizing~$L$ (if it exists).  It is
not difficult to adjust the DP to also compute an optimal tangle.  Let
$\lambda = \prod_{i<j} (l_{ij}+1)$ be the number of distinct sublists
of~$L$.  The DP computes a table~$H$ of size $\lambda$ where, for any
sublist~$L'$ of~$L$, $H(L')$ is the optimal height of a tangle
realizing~$L'$ if $L'$ is feasible (otherwise $H(L')=\infty$).
We compute the entries of~$H$ in non-decreasing order
of list length.  Let~$L'$ be the next list to consider.  We first
check the consistency of~$1(L')$ by computing the map $\id_n 1(L')$ in
$\oh(n^2)$ time (according to
Proposition~\ref{prop:feasibility-simple}).

If $1(L')$ is not consistent, then $L'$ is not consistent, %
and we set
$H(L')=\infty$.  Otherwise, we compute the optimal height of~$T(L')$
of~$L'$ (if~$L'$ is feasible).  To this end, let
$\rho = \id_n 1(L') = \id_n L'$.  This is the final permutation
of any tangle that realizes~$L'$.  Now we go through the set $S_\rho$
of permutations that are adjacent to~$\rho$.  According to
\cref{lem:adjacent}, $|S_\rho|=F_{n+1}-1$.  For each
permutation~$\pi \in S_\rho$, the set $L(\langle\pi,\rho\rangle)$ is
the set of disjoint swaps that transforms~$\pi$ into~$\rho$.  Then

\begin{equation}
  \label{eq:HofLprime}
  H(L') = \min_{\pi \in S_\rho}
  H(L'-L(\langle\pi,\rho\rangle))+1.
\end{equation}
Note that, if~$L'$ is not feasible, then, for every $\pi \in S_\rho$,
$H(L'-L(\langle\pi,\rho\rangle))=\infty$ and hence $H(L')=\infty$.

After we have filled the whole table~$H$, the entry~$H(L)$ contains the
desired height of a realization of~$L$ (or $\infty$ if $L$ is not
feasible).

Computing the minimum in \cref{eq:HofLprime}
needs $F_{n+1}-1$ lookups in the table~$H$.
In total we spend $\oh(\lambda(n^2 + (F_{n+1}-1)n \log |L|)) =
\oh(\lambda(F_{n+1}-1)n \log |L|)$ time for filling the table~$H$.
The $(\log|L|)$-factor is due to the fact that we need to modify
($\oh(n)$) entries of $L'$ in order to obtain
$L''=L'-L(\langle\pi,\rho\rangle)$ and to look up the table
entry~$H(L'')$.

Assuming that $n\ge 2$, we bound $\lambda$ as follows, where
we obtain the first inequality
from the inequality between arithmetic and geometric means,
and the second one from Bernoulli's inequality.
\[\lambda
=\prod_{i<j} (l_{ij}+1)
\le \left(\frac{\sum_{i<j} (l_{ij}+1)}{{n\choose 2}}\right)^{n\choose 2}
=\left(\frac{|L|}{{n\choose 2}}+1\right)^{n\choose 2}
\le \left(\frac{2|L|}{n^2}+1\right)^{n^2/2}.\]
\end{proof}

Note that, since $1+x\le e^x$ for any $x \in \mathbb{R}$,
the running time $\oh((2|L|/n^2+1)^{n^2/2} \cdot \varphi^n \cdot n \cdot \log|L|)$
of the DP is upper-bounded by
$\oh(e^{|L|} \cdot \varphi^n \cdot n \cdot \log|L|)$.

\section{Algorithms for Deciding Feasibility}
\label{sec:fpt}

We investigate the feasibility problem for different types of lists
in this section. First we consider general lists, then simple lists,
\emph{odd} lists and, finally, \emph{even} lists.
A list
$L = (l_{ij})$ is \emph{even} if all $l_{ij}$ are even, and \emph{odd} if
all non-zero $l_{ij}$ are
odd.
Note that $L$ is even if and only if the list $1(L)$ is the zero
list (recall that $1(L)=(l_{ij} \bmod 2)$),
and $L$ is odd if and only if $1(L)=2(L)$.

For any list to be feasible, each triple of wires $i < j < k$
requires an $(i,j)$ or a $(j,k)$ swap if there is an $(i,k)$
swap~-- otherwise wires~$i$ and $k$ would be separated by wire~$j$ in
any tangle.  We call a list fulfilling this property
\emph{non-separable}.
For odd lists, non-separability is implied by consistency (because
consistency is sufficient for feasibility, see
Proposition~\ref{prop:odd-list-feasibility} in the following).
The NP-hardness reduction from \cref{sec:complexity} shows that a non-separable
list can fail to be feasible even when it is consistent.

\paragraph{General lists.}

Our DP for \tangleMinimization runs in
$\oh((2|L|/n^2+1)^{n^2/2}\cdot\varphi^n\cdot n\log|L|)$ time.
We adjust this
algorithm to the task of testing feasibility, which makes the
algorithm simpler and faster.  Then we bound the entries of
minimal feasible lists (defined in \cref{sec:intro}) and use this bound to turn our
exact algorithm into a fixed-parameter algorithm where the parameter
is the number of wires (i.e.,~$n$).

\begin{theorem}
	\label{thm:exact}
	There is an algorithm that, given a list~$L$ of order~$n$, tests
	whether~$L$ is feasible in
	$O\big((2|L|/n^2+1)^{n^2/2}\cdot n^3\cdot \log |L|\big)$ time.
\end{theorem}

\begin{proof}
	Let~$F$ be a Boolean table with one entry for each sublist~$L'$ of~$L$
	such that $F(L')=\texttt{true}$ if and only if~$L'$ is feasible.
	This table can be filled by means of a dynamic programming
	recursion.  The empty list is feasible.  Let~$L'$ be a sublist
	of~$L$ with~$|L'|\ge 1$ and assume that for each strict sublist
	of~$L'$, the corresponding entry in~$F$ has already been determined.
	A sublist~$\tilde L$ of~$L$
	is feasible if and only if there is a
	realizing tangle of~$\tilde L$ of height~$|\tilde L|+1$.  For each
	$(i,j)$ swap in~$L'$, we check if there is a tangle realizing~$L'$ of
	height~$|L'|+1$ such that $(i,j)$ is the last swap.  If no such swap
	exists, then~$L'$ is infeasible, otherwise it is feasible.  To
	perform the check for a particular $(i,j)$ swap, we consider the strict
	sublist~$L''$ of~$L'$ that is identical to~$L'$ except an $(i,j)$
	swap is missing.  If $F(L'')=\texttt{true}$, we
	compute the final positions of~$i$ and~$j$ with respect to~$L''$ (see \cref{sec:algorithms}).
	The desired tangle exists if and only if these positions differ by
	exactly one.
	
	The number of sublists of~$L$ is upper bounded by
	$(2|L|/n^2+1)^{n^2/2}$, see the proof of \cref{thm:exact-general}.
	For each sublist,
	we have to check~$\oh(n^2)$ swaps.  To check a swap, we have to
	compute the final positions of two wires, which can be done
	in~$\oh(n\log|L|)$ time.
\end{proof}

Next we introduce a tangle-shortening construction.
We use the following lemma that follows from odd-even sort and is
well known \cite{k-acp3-98}.

\begin{lemma}
	\label{lem:connecting}
	For each integer $n \ge 2$ and
	each pair of permutations $\pi,\sigma\in S_n$, we can construct
	in $\oh(n^2)$ time a tangle $T$ of height
	at most $n+1$ that starts with $\pi$, ends in $\sigma$, and whose
	list $L(T)$ is simple.
\end{lemma}

We use \cref{lem:connecting} in order to shorten lists without
changing their type.

\begin{lemma}\label{lem:shortening}
	Let $T = \langle \pi_1, \pi_2 \dots,\pi_h \rangle$ be a tangle
	with $L = L(T) = (l_{ij})$,
	and let $\{\pi_1, \pi_h\} \subseteq P \subseteq \{\pi_1, \pi_2, \dots,\pi_h\}$.
	If, for every $1\le i<j\le n$ and
	$l_{ij}>0$, there exists a permutation $\pi \in P$
	with $\pi(j)<\pi(i)$,
    then we can construct a tangle
$T'$ with $L' = L(T') = (l'_{ij})$ such that
$l'_{ij}\le \min\{l_{ij},|P|-1\}$
and $2(L') = 2(L)$.
\end{lemma}

\begin{proof}
	We construct the tangle $T'$ as follows.  Let $p=|P|$ and
	consider the elements of $P$ in the order of first occurrence
        in~$T$, that is, $P=\{\pi'_1, \pi'_2, \dots, \pi'_p\} =
        \{\pi_{m_1},\pi_{m_2},\dots,\pi_{m_p}\}$ such that
        $1=m_1 < m_2 < \dots < m_p=h$.
	For every two consecutive elements~$\pi'_k$ and $\pi'_{k+1}$
	with $k \in [p-1]$, we create a	tangle $T'_k$ that starts from
	$\pi'_k$, ends at~$\pi'_{k+1}$, and whose list
	$L(T'_k)=(l'_{k,ij})$ is simple. Note that, by
	\cref{lem:connecting}, we can always construct such a tangle
	of height at most $n+1$, where $n$ is the number of wires in $T$.
	Now let $T' = \langle T'_1,T'_2,\dots,T'_{p-1} \rangle$.
        For each $k\in [p-1]$, let~$T_k$ be the subtangle of~$T$ that
        starts at~$\pi'_k$ and 
	ends at $\pi'_{k+1}$. Note that both these permutations are in $T$.
	Let $L(T_k)=(l_{k,ij})$. The simplicity of
	the list $L(T'_k)$ implies that, for every
	$i,j\in [n]$, $l'_{k,ij}\le l_{k,ij}$.
	Then, for every $i,j\in [n]$, $l'_{ij}\le \min\{l_{ij},p-1\}$.
	Since the tangles $T'$ and $T$ have common initial and final permutations,
	for every $i,j \in [n]$, the numbers $l_{ij}$ and $l'_{ij}$ have
	the same parity, that is, $1(L')=1(L)$.
	Since for every $1\le i<j\le n$ and	$l_{ij}>0$, there exists a
	permutation $\pi \in P$	with $\pi(j)<\pi(i)$, it means
	that if $l_{ij}>0$, then $l'_{ij}>0$.
	Hence, $2(L') = 2(L)$.
\end{proof}

We want to upperbound the entries of a minimal feasible list.  We
first give a simple bound, which we then improve by a factor of roughly~2.

\begin{proposition}
	\label{prop:simple-upper-bound}
	If $L=(l_{ij})$ is a minimal feasible list of order~$n$, then
	$l_{ij}\le {n \choose 2}+1$ for each $i,j\in [n]$.
\end{proposition}
\begin{proof}
	The list $L$ is feasible, so there is a tangle~$T=\langle
        \pi_1,\pi_2,\dots,\pi_h \rangle$ realizing~$L$.
        We construct a set $P$ of permutations of~$T$, starting
        with $P = \{\pi_1,\pi_h\}$.  Then, for each pair $(i,j)$ with
        $1 \le i < j \le n$ and $l_{ij}\ge 1$, we add to~$P$ a
        permutation~$\pi_{ij}$ from~$T$ such that
        $\pi_{ij}(j)=\pi_{ij}(i)+1$.  Note that $T$ must contain such
        a permutation since we assume that $l_{ij}\ge 1$ and that
        $\pi_1=\id_n$.  Clearly, $|P| \le {n \choose 2}+2$.
	
	Now \cref{lem:shortening} yields a tangle $T'$
	with $L'=L(T')=(l'_{ij})$ such that
	$l'_{ij}\le \min\{l_{ij},|P|-1\}$ and $2(L') = 2(L)$.
	Hence, $L'\to L$.
	The list~$L$ is minimal, so $L=L'$. Thus, since
	$|P| \le {n \choose 2}+2$,
	$l_{ij}=l'_{ij}\le |P|-1\le {n \choose 2}+1$ for each
	$i,j\in[n]$.
\end{proof}

Now we improve this bound by a factor of roughly 2.

\begin{proposition}
	\label{prop:minimal-feasible-lists-upper-bound}
	If $L=(l_{ij})$ is a minimal feasible list of order~$n$, then
	$l_{ij}\le n^2/4+1$ for each $i,j\in [n]$.
\end{proposition}

\begin{proof}
	The list $L$ is feasible, so there is a tangle~$T=\langle
	\pi_1,\pi_2,\dots,\pi_h \rangle$ realizing~$L$.
    We again construct a set $P$ of permutations of~$T$, starting
	from $P = \{\pi_1,\pi_h\}$ as follows.

    Let $G$ be the graph with vertex set $[n]$ that has an edge
    for each pair $\{i,j\}$ with $l_{ij}\ge 1$.  The edge $ij$, if
    it exists, has weight $|i-j|$.  For an edge $ij$ with $i<j$, we
    say that a permutation~$\pi$ of~$T$ \emph{witnesses} the edge~$ij$
    of~$G$ (or the swap $(i,j)$ of~$L$) if $\pi(j)<\pi(i)$.

    We repeat the following step until every edge of~$G$ is colored.
    Pick any non-colored edge of maximum weight and color it red.
    Let $i$ and $j$ with $i<j$ be the endpoints of this edge.
    Since $l_{ij}\ge 1$ and the tangle~$T$ realizes the list~$L$,
    $T$ contains a permutation~$\pi$ with $\pi(j)<\pi(i)$.
    Note that $\pi$ witnesses the edge~$ij$.  Add~$\pi$ to~$P$.
    Now for each~$k$ with $i<k<j$, do the following.
    Since $\pi_1 = \id_n$ and $\pi(j)<\pi(i)$, it
    holds that $\pi(k)<\pi(i)$ or $\pi(j)<\pi(k)$ (or both).
    In other words, $\pi$ witnesses the edge~$ki$, the edge~$jk$,
    or both.  We color each witnessed edge blue.
	
    The coloring algorithm ensures that the graph $G$ has no red
    triangles, so, by Tur\'an's theorem~\cite{turan1941external}, $G$
    has at most $n^2/4$ red edges.  Hence $|P| \le 2+ n^2/4$.  By
    construction, every edge of~$G$ is witnessed by a permutation
    in~$P$.  Thus, \cref{lem:shortening} can be applied to~$P$.
    This yields a tangle~$T'$ with $L'=L(T')=(l'_{ij})$ such that
	$l'_{ij}\le \min\{l_{ij},|P|-1\}$ and $2(L') = 2(L)$.
	Hence, $L'\to L$.  The list $L$ is minimal, therefore $L'=L$.
	Since each entry of the list~$L'$ is at most $|P|-1\le n^2/4+1$,
	the same holds for the entries of the list $L$.
\end{proof}

Combining Proposition~\ref{prop:minimal-feasible-lists-upper-bound}
and our exact algorithm from \cref{thm:exact} yields a fixed-parameter
tractable algorithm with respect to~$n$.

\begin{theorem}
	\label{thm:fpt}
	There is a fixed-parameter algorithm for \listFeasibility
	with respect to the parameter~$n$.  Given a
	list~$L$ of order~$n$, the algorithm tests whether~$L$ is feasible
	in $O\big((n/2)^{n^2}\cdot n^3 \log n + n^2 \log |L| \big)$ time.
\end{theorem}

\begin{proof}
	Given the list~$L=(l_{ij})$, let $L'=(l'_{ij})$ with
	$l'_{ij}=\min \{ l_{ij}, n^2/4+1 \}$ for each $i,j\in[n]$.
	We use our exact algorithm described in the proof of
	\cref{thm:exact} to check whether the list~$L'$ is feasible.
	Since our algorithm checks the feasibility of every sublist~$L''$
	of~$L'$, it suffices to combine this with checking whether
	$2(L'')=2(L)$.  If we find a feasible sublist $L''$ of the same
	type as~$L$, then, by
	Proposition~\ref{prop:minimal-feasible-lists-upper-bound},
	$L$ is feasible; otherwise, $L$ is infeasible.  Checking the type of~$L''$
	is easy.  The runtime for this check is dominated by the runtime for
	checking the feasibility of~$L''$.  Constructing the list~$L'$ takes
	$\oh(n^2 \log |L|)$ time.
	Note that $|L'| \le {n \choose 2} \cdot (n^2/4+1) \le (n^4-4n^2)/8$.  Plugging
	this into the runtime $O\big((2|L'|/n^2+1)^{n^2/2}\cdot n^3 \log|L'|\big)$ of
	our exact algorithm (\cref{thm:exact}) yields a total
	runtime of $O\big((n/2)^{n^2}\cdot n^3 \log n+ n^2 \log |L| \big)$.
\end{proof}

\paragraph{Simple Lists.}
If we restrict our study to simple lists, we can easily decide feasibility.

\begin{proposition}\label{prop:feasibility-simple}
	A simple list $L$ is feasible if and only if
	$L$ is consistent. Thus, we can check
	the feasibility of $L$ in $\oh(n+|L|)$ time.
\end{proposition}

\begin{proof}
	Clearly, if $L$ is feasible, then~$L$ is also consistent.
	If $L$ is consistent, then $\id_n L$ is a permutation.
	By \cref{lem:connecting}, there exists a tangle $T$ which starts from~$\id_n$, ends at
	$\id_n L$, and the list $L(T)$ is simple. By \cref{lem:final-independence}, $\pi L(T)=\pi L$.
	By \cref{lem:Lpi1}, $L(T)=L$.
	So $L$ is also feasible. We can check the
	consistency of~$L$ in $\oh(n+|L|)$ time, which is equivalent
	to checking the feasibility of~$L$.
\end{proof}

\paragraph{Odd Lists.}
For odd lists, feasibility reduces to that of simple lists.  For
$A\subseteq [n]$, let $L_A$ be the list that consists of all swaps
$(i,j)$ of $L$ such that $i, j \in A$.

\begin{proposition}
	\label{prop:odd-list-feasibility}
	For $n \ge 3$ and an odd list $L$, the following statements are equivalent:
	\begin{enumerate}
		\item The list $L$ is feasible.
		\item The list $1(L)$ is feasible.
		\item For each triple $A\subseteq [n]$, the list $L_A$ is feasible.
		\item For each triple $A\subseteq [n]$, the list $1(L_A)$ is feasible.
		\item The list $L$ is consistent.
		\item The list $1(L)$ is consistent.
		\item For each triple $A\subseteq [n]$, the list $L_A$ is consistent.
		\item For each triple $A\subseteq [n]$, the list $1(L_A)$ is consistent.
	\end{enumerate}
\end{proposition}

\begin{proof}
	We prove the proposition by proving three cycles of implications
	$1\Rightarrow 5\Rightarrow 6\Rightarrow 2\Rightarrow 1$,
	$3\Rightarrow 7\Rightarrow 8\Rightarrow 4\Rightarrow 3$,
	and $1\Rightarrow 3\Rightarrow 2\Rightarrow 1$.
	
	$1\Rightarrow 5$. Clearly, all feasible lists are consistent.
	
	$5\Rightarrow 6$. Consistency of~$L$ means that $\id_n L\in S_n$.
	Since $\id_n 1(L)=\id_n L$, the list $1(L)$ is consistent, too.
	
	$6\Rightarrow 2$. Follows from Proposition~\ref{prop:feasibility-simple} because the list $1(L)$ is simple.
	
	$2\Rightarrow 1$. We decompose $L$ into $1(L)$ and $L'=(L -
	1(L))$.
        Note that $L'=(l'_{ij})$ is an even list.  Let $(i,j) \in
	L'$.  Then $(i,j) \in 1(L)$ because $L$ is odd.  Consider a
	tangle~$T$ realizing $1(L)$.  Let $\pi$ be the layer in~$T$ where
	the swap~$(i,j)$ occurs.  Behind~$\pi$, insert $l_{ij}$ new layers such
	that the difference between one such layer and its previous layer
	is only the swap~$(i,j)$.  Observe that every second new layer
	equals~$\pi$~-- in particular the last one, which means that we
	can continue the tangle with the remainder of~$T$.  Applying this
	operation to all swaps in~$L'$ yields a tangle realizing~$L$.
	
	$3\Rightarrow 7$. Clearly, all feasible lists are consistent.
	
	$7\Rightarrow 8$. Follows from the equality $\id_n 1(L_A)=\id_n L_A$.
	
	$8\Rightarrow 4$. Follows from Proposition~\ref{prop:feasibility-simple}, because the list $1(L_A)$ is simple.
	
	$4\Rightarrow 3$. For every triple $A\subseteq [n]$, we can argue
	as in the proof~$(2\Rightarrow 1)$.
	
	$1\Rightarrow 3$. Trivial.
	
	$3\Rightarrow 2$.  Let $1\le i<k<j\le n$.  By the equivalence
	$(1\Leftrightarrow 2)$, the odd list $L_{\{i,k,j\}}$ is infeasible
	if and only if $1(L_{\{i,k,j\}})$ is infeasible, that is, either
	$(i,j) \in L$ and $(i,k)$, $(k,j) \not\in L$, or
	$(i,j) \not\in L$ and $(i,k)$, $(k,j) \in L$.
	Define a binary relation $\le '$ on the set $[n]$ by letting
	$i\le'j$ if and only if either $i\le j$ and $(i,j) \not\in L$, or
	$i>j$ and $(i,j) \in L$.
        Using the feasibility of~$L_A$ for all triples $A\subseteq [n]$,
	it follows that $\le'$ is a linear order.
        Let $\pi$ be the (unique) permutation of the set
	$[n]$ such that $\pi^{-1}(1)\le'\pi^{-1}(2)\le'\dots\le' \pi^{-1}(n)$.
	Observe that $L(\pi)=1(L)$, so the list $1(L)$ is feasible.
\end{proof}

Note that, for any feasible list~$L$, it does not necessarily hold
that $2(L)$ is feasible; see, e.g., the list~$L_n$ from
Observation~\ref{obs:unique}.

\paragraph{Even Lists.}
An even list is always consistent since it
does not contain an odd number of swaps and the final permutation is the same as the initial one.
We show that, for sufficiently ``rich'' lists,
non-separability is sufficient for an even list
to be feasible, but in general this is not true.

\begin{proposition}
	\label{prop:big-even-L-are-realizable}
	Every non-separable even list~$L=(l_{ij})$ with $l_{ij}\ge n$ or $l_{ij} = 0$ for every $1 \leq i, j \leq n$
	is feasible.
\end{proposition}

\begin{proof}
	We define a binary relation $\le_L$ on the set of wires $[n]$ for each $i,j \in [n]$ as follows.
	We set $i\le_L j$ if and only if $i\le j$ and
	$l_{ij}=0$. Since the list $L$ is non-separable, the relation $\le_L$ is
	transitive, so it is a partial order. The dimension $d$ of a partial order on the set $[n]$ is at
	most $\lceil\frac n2\rceil$~\cite{h-do-SRKU55}, this is,
	there exist $d$ \emph{linear orders} $\le_1,\dots, \le_d$ of the set $[n]$ such that for each $i,j\in
	[n]$ we have $i\le_L j$ if and only if $i\le_t j$ for each $t \in [d]$.
	So $\le_L$ can be seen as the intersection of $\le_1,\dots, \le_d$.
	For each linear order $\le_k$ with $k \in [d]$, let
	$\pi_k$ be the (unique) permutation of the set $[n]$ such that
	$\pi^{-1}_k(1)\le_k\pi^{-1}_k(2)\le_k\dots\le_k \pi^{-1}_k(n)$ and $L_k=L(\pi_k)$.
	As a consequence of \cref{lem:Lpi1}, the list $L_k$ is feasible.
	Let $S_k=L_1+L_1+L_2+L_2+\dots+L_k+L_k$. Observe that the list $S_k$ is even. So
	$\id_n S_k=\id_n$ and $\id_n (S_k-L_k)=\id_n L_k=\id_n L(\pi_k)=\pi_k$.
	Since $\pi_k L_k=\id_n$, we can inductively show for each $k$ that $S_k$ is feasible.
	Therefore, the list $L'=(l'_{ij})=S_d=L_1+L_1+L_2+L_2+\dots+L_d+L_d$ is feasible.
	Let $1\le i<j\le n$. If $l'_{ij}=0$ then for
	all $k\in [d]$ it holds that $\pi_k(i)<\pi_k(j)$, hence, $i\le_k j$, which means $i\le_L j$ and $l_{ij}=0$.
	On the other hand, if
	$l'_{ij}\ne 0$ then $l'_{ij} \le 2d \le n \le l_{ij}$.
	We can extend a tangle~$T'$ realizing $L'$ such that we execute the remaining (even)
	number of $l_{ij} - l'_{ij}$ swaps of the wires $i$ and $j$ for each non-zero entry
	of $L$ after an execution of an $(i,j)$ swap in $T'$.
	Thus, the feasibility of $L$ follows from the feasibility of $L'$.
\end{proof}

In the following we give an example of non-separable lists that
are not feasible.
Note that any triple $A\subseteq [n]$ of an even list is feasible if
and only if it is non-separable (which is not true for general lists,
e.g., the list $L=\{12,23\}$ is not feasible).

\begin{proposition}\label{prop:non-separ-even}
	Not every non-separable even list $L$ is feasible.
\end{proposition}

We construct a family $(L^\star_m)_{m\ge1}$ of non-separable lists
whose entries are all zeros or twos such that $L^\star_m$ has $2^m$
wires and is not feasible for $m \ge 4$.  Slightly deviating from our
standard notation, we number the wires od~$L^\star_m$ from~0 to
$2^m-1$.  There is no swap between two wires $i<j$ in $L^\star_m$ if
each~1 in the binary representation of $i$ also belongs to the
binary representation of $j$, that is,
the bitwise OR of $i$ and $j$ equals $j$; otherwise, there are
two swaps between $i$ and~$j$. E.g., for $m = 4$, wire
$1=0001_2$ swaps twice with wire $2=0010_2$, but doesn't swap with
wire $3=0011_2$.

Each list $L^\star_m$ is clearly non-separable: assume that there
exists a swap between two
wires $i=(i_1 i_2 \dots i_m)_2$ and $k=(k_1 k_2 \dots k_m)_2$ with $k>i+1$.
Then there has to be some index $a$ with $i_a=1$ and $k_a=0$.
Consider any $j=(j_1 j_2 \dots j_m)_2$ with $i<j<k$.
By construction of $L^\star_m$, if $j_a=0$, then there are two swaps
between $i$ and $j$;
if $j_a=1$, then there are two swaps between $j$ and $k$.

We used two completely different computer programs\footnote{%
  One program is based on combining realizations for triplets of
  wires~\cite{github2022}; the other is based on a SAT
  formulation~\cite{a-practical-22}.  Both implementations are
  available on github.}  to verify that $L^\star_4$~-- and hence every
list~$L^\star_m$ with $m \ge 4$~-- is infeasible.  Unfortunately, we could
not find a combinatorial proof showing this.  The list~$L^\star_m$ has
$\frac{1}{2}\sum_{r=1}^{m} 3^{r-1} 2^{m-r}(2^{m-r}-1)$ swaps of
multiplicity~2, so $L^\star_4$ has 55 distinct swaps.  The full list
$L^\star_4$ in matrix form is given below.

\begingroup
\definecolor{lightgray}{rgb}{0.6,0.6,0.6}
\newcommand{\x}{\cellcolor{lightgray}2}
\setcounter{MaxMatrixCols}{16}
\setlength\arraycolsep{6pt}
\[L^\star_4=\begin{pmatrix}
0 & 0 & 0 & 0 & 0 & 0 & 0 & 0 & 0 & 0 & 0 & 0 & 0 & 0 & 0 & 0 \\
0 & 0 & \x& 0 & \x& 0 & \x& 0 & \x& 0 & \x& 0 & \x& 0 & \x& 0 \\
0 & 0 & 0 & 0 & \x& \x& 0 & 0 & \x& \x& 0 & 0 & \x& \x& 0 & 0 \\
0 & 0 & 0 & 0 & \x& \x& \x& 0 & \x& \x& \x& 0 & \x& \x& \x& 0 \\
0 & 0 & 0 & 0 & 0 & 0 & 0 & 0 & \x& \x& \x& \x& 0 & 0 & 0 & 0 \\
0 & 0 & 0 & 0 & 0 & 0 & \x& 0 & \x& \x& \x& \x& \x& 0 & \x& 0 \\
0 & 0 & 0 & 0 & 0 & 0 & 0 & 0 & \x& \x& \x& \x& \x& \x& 0 & 0 \\
0 & 0 & 0 & 0 & 0 & 0 & 0 & 0 & \x& \x& \x& \x& \x& \x& \x& 0 \\
0 & 0 & 0 & 0 & 0 & 0 & 0 & 0 & 0 & 0 & 0 & 0 & 0 & 0 & 0 & 0 \\
0 & 0 & 0 & 0 & 0 & 0 & 0 & 0 & 0 & 0 & \x& 0 & \x& 0 & \x& 0 \\
0 & 0 & 0 & 0 & 0 & 0 & 0 & 0 & 0 & 0 & 0 & 0 & \x& \x& 0 & 0 \\
0 & 0 & 0 & 0 & 0 & 0 & 0 & 0 & 0 & 0 & 0 & 0 & \x& \x& \x& 0 \\
0 & 0 & 0 & 0 & 0 & 0 & 0 & 0 & 0 & 0 & 0 & 0 & 0 & 0 & 0 & 0 \\
0 & 0 & 0 & 0 & 0 & 0 & 0 & 0 & 0 & 0 & 0 & 0 & 0 & 0 & \x& 0 \\
0 & 0 & 0 & 0 & 0 & 0 & 0 & 0 & 0 & 0 & 0 & 0 & 0 & 0 & 0 & 0 \\
0 & 0 & 0 & 0 & 0 & 0 & 0 & 0 & 0 & 0 & 0 & 0 & 0 & 0 & 0 & 0
\end{pmatrix}\]
\endgroup

\section{Open Problems}

Obviously it would be interesting to design faster algorithms for
\tangleMinimization and \listFeasibility.
In particular, for the special case of simple lists, our
exact algorithm running in $\oh(n!\varphi^n)$ time and the algorithm
of Baumann~\cite{b-hmst-BTh20} running in $\oh(n!\psi^n)$ time
(where $\varphi \approx 1.618$ and $\psi \approx 1.325$) are
not satisfactory given that
odd-even sort~\cite{si-fectbs-TCS87} can compute a solution of height at
most one more than the optimum in $\oh(n^2)$ time.
This leads to the question whether
height-minimization is NP-hard for simple lists.
For general lists, one can potentially obtain a faster algorithm for
\listFeasibility by improving the upper bound for entries of
minimal feasible lists
(see Proposition~\ref{prop:minimal-feasible-lists-upper-bound}
for the current upper bound).

Another research direction is to devise approximation
algorithms for \tangleMinimization and \listFeasibility.

\paragraph*{Acknowledgments.}
We thank Thomas C.\ van Dijk for stimulating discussions in the
initial phase of this work.
We thank Stefan Felsner for discussions
about the complexity of \listFeasibility.
We thank the anonymous reviewers of earlier
versions of this paper for helpful comments.

\bibliographystyle{plainurl}
\bibliography{abbrv,tangles}

\end{document}